\documentclass[a4paper,onecolumn,11pt,accepted=2026-05-20]{quantumarticle}
\pdfoutput=1
\usepackage[english]{babel}
\usepackage{amsmath}
\usepackage{amsthm}
\usepackage{amssymb}
\usepackage{mathrsfs}  
\usepackage{fullpage}
\usepackage[numbers,sort&compress]{natbib}
\usepackage{comment}
\usepackage[colorlinks=true]{hyperref}
\usepackage{braket}
\usepackage{mathtools}
\usepackage{bbm}
\newtheorem{thm}{Theorem}

\newtheorem*{thm*}{Theorem}
\newtheorem{algorithm}{Algorithm}

\newtheorem{cor}{Corollary}

\newtheorem{lem}{Lemma}

\newtheorem{prop}{Proposition}

\newtheorem*{prop*}{Proposition}
\theoremstyle{definition}
\newtheorem{defn}{Definition}

\theoremstyle{remark}
\newtheorem{rem}{Remark}

\numberwithin{equation}{section}
\newtheorem{ass}{Assumption}

\def\norma #1{\left\lVert #1 \right\rVert}

\def\N{{\mathbb N}}

\def\R{{\mathbb R}}
\def\CC{{\mathbb C}}
\def\L{{\mathcal L}}

\def\E{{\mathbb E}}

\def\X{{\mathcal X}}

\def\norma #1{\left\lVert #1 \right\rVert}

\def\de{{\rm d}}

\def\K{{\mathcal K}}

\def\lin{{\rm lin}}

\usepackage[dvipsnames]{xcolor}
\newcommand{\modifica}[1]{#1}               

\makeatletter
\@namedef{subjclassname@2020}{%
  \textup{2020} Mathematics Subject Classification}
\makeatother

\title{Efficient classical computation of the neural tangent kernel of quantum neural networks}
\author[A.~Melchor Hernandez]{Anderson Melchor Hernandez}
\affiliation{Dipartimento di Matematica,
Università di Bologna, Piazza di Porta San Donato 5, 40126, Bologna (Italy)}
\orcid{0000-0003-4437-615X}
\email{anderson.melchor@unibo.it}

\author[D.~Pastorello]{Davide Pastorello}
\affiliation{Dipartimento di Matematica,
Università di Bologna, Piazza di Porta San Donato 5, 40126, Bologna (Italy)\\TIFPA-INFN, via Sommarive 14, 38123 Povo (Trento), Italy}
\orcid{0000-0001-5915-6796}
\email{davide.pastorello3@unibo.it}

\author[G.~De Palma]{Giacomo De Palma}
\affiliation{Dipartimento di Matematica,
Università di Bologna, Piazza di Porta San Donato 5, 40126, Bologna (Italy)}
\orcid{0000-0002-5064-8695}
\email{giacomo.depalma@unibo.it}


\begin{document}
\begin{abstract}
We propose an efficient classical algorithm to estimate the Neural Tangent Kernel (NTK) associated with a broad class of quantum neural networks. These networks consist of arbitrary unitary operators belonging to the Clifford group interleaved with parametric gates given by the time evolution generated by an arbitrary Hamiltonian belonging to the Pauli group.
The proposed algorithm leverages a key insight: the average over the distribution of initialization parameters in the NTK definition can be exactly replaced by an average over just four discrete values, chosen such that the corresponding parametric gates are Clifford operations. This reduction enables an efficient classical simulation of the circuit.
Combined with recent results establishing the equivalence between wide quantum neural networks and Gaussian processes [Girardi \emph{et al.}, Comm. Math. Phys. 406, 92 (2025); Melchor Hernandez \emph{et al.}, \modifica{Ann. Henri Poincar{\'e} (2025)}], our method enables efficient computation of the expected output of wide, trained quantum neural networks, and therefore shows that such networks cannot achieve quantum advantage.
\end{abstract}

\maketitle

\tableofcontents

\section{Introduction}

Quantum Machine Learning (QML) is an interdisciplinary field that merges the principles of quantum computing with classical machine learning techniques \cite{de2019primer,schuld2015, pastorello2023concise}. One of the core ideas of QML is to harness quantum algorithms and the unique properties of quantum mechanics such as superposition, entanglement, and quantum parallelism to enhance the performance of deep neural models \cite{biamonte2017}. Quantum neural networks constitute the quantum version of deep neural models. These new models are based on quantum circuits and generate functions given by the expectation values of a quantum observable measured on the output of a quantum circuit made by parametric gates \cite{girardi2024,schuld2018}. These parameters are typically optimized by gradient descent, which involves iterative adjustment to minimize a cost function and improve the performance of the quantum circuit in the processing and analysis of data \cite{schuld2021effect}.
Several works have focused on the analysis of quantum neural networks, as it is believed that they can combine the computational power of quantum computers with the capabilities of deep learning algorithms \cite{lloyd2020quantum}. Ref. \cite{liu2021} shows that an exponential quantum speed-up can be obtained via the use of a quantum-enhanced feature space, where each data point is mapped in a non-linear way to a quantum state, and then classified by a linear classifier in a high-dimensional Hilbert space \cite{havlivcek2019}. Nevertheless, a significant disadvantage lies in the need to determine the appropriate parameters to configure the quantum circuit beforehand, and it is not yet clear whether these parameters can be effectively obtained using a variational technique \cite{cinelli2021var}.

A rigorous mathematical characterization of the training dynamics of quantum neural networks becomes possible in the limit of infinite width. Ref. \cite{abedi2023} proves that in the limit of infinite width with the depth kept fixed, trained quantum neural networks with constant depth operate in the lazy regime (\emph{i.e.}, the maximum amount by which the training can change a parameter tends to zero) and are capable of perfectly fitting the training examples. However, such a regime is not interesting for quantum advantage, since at finite depth, the past light-cone of each measured qubit has a finite size and the expectation value of the measured observable can be efficiently estimated on a classical computer.
The recent works \cite{girardi2024,melchor2025quantitative} consider the joint limit of infinite width and depth.
Under the hypothesis that the depth grows at most logarithmically with respect to the number of qubits, they prove that the probability distribution of the trained model function converges in distribution to a Gaussian process. The key element of the proof is showing that the training happens in the lazy regime, and therefore the dependence of the model function on the parameters can be approximated by its linearized version near the initialization values. Consequently, in the limit the model becomes linear and the training has an analytic solution whose probability distribution is Gaussian with analytically computable mean and covariance.
Such a linear model is governed by a kernel called the Neural Tangent Kernel (NTK), which has been investigated in several recent studies \cite{schuld2015, Shirai:2021jpy, PRXQuantum.3.030323, Yu_2024, scala2025towards}.

Let us recall that, given a general parametric model $f_\theta:\mathcal X\rightarrow\mathbb{R}$, with $\theta\in\mathbb{R}^L$, where $\mathcal X$ is the feature space, the {\em empirical} NTK is defined as \cite{NEURIPS2018_5a4be1fa}:
\begin{equation}\label{eq:defNTK}
    \hat K_\theta(x,x')= \nabla_\theta f_\theta(x)^T\nabla_\theta f_\theta(x')\qquad x,x'\in\mathcal X,
\end{equation}
and is the central object to describe the training dynamics of the model generated by the time evolution of the parameters under gradient flow.
The {\em analytic} NTK is defined by the expectation value of the empirical NTK with respect to the random sampling of the parameters at initialization:
\begin{equation}
    K(x,x')=\mathbb E_\theta\,\hat K_\theta(x,x')\,,
\end{equation}
and governs the training dynamics of the model in the limit of infinite width.

\subsection{Our contribution}
In this work, we provide a classical algorithm to efficiently estimate the NTK for a very wide class of quantum neural networks and present quantitative results on both the sample and computational complexity of the proposed algorithm to show the efficiency of the estimation.
We then leverage the proposed algorithm to efficiently estimate the average of the model functions of very wide trained quantum neural networks.
We consider a finite set $\mathcal X$ of possible inputs and a model function defined by
\begin{equation}\label{model}
    f_\theta(x)=\langle 0^n|U_{x,\theta}^\dagger\, O\, U_{x,\theta}|0^n\rangle,
\end{equation}
where $n$ denotes the number of qubits, $|0^n\rangle = |0\rangle^{\otimes n}$, $O$ is the measured $n$-qubit observable, and $U_{x,\theta}$ is a parametric quantum circuit made by nonparametric unitary operations belonging to the Clifford group, which can depend on the input $x$ in an arbitrary way, interleaved with parametric gates given by the time evolution generated by a Hamiltonian belonging to the Pauli group (see \eqref{circuit}), so that each entry of the parameter vector is a rotation angle in $[0,2\pi)$.

We provide an efficient classical algorithm to estimate the NTK of any such a model function.
Our estimator is unbiased, meaning that its expected value coincides with the true kernel. Our first contribution can be stated informally in the following way:

\begin{thm}[Informal statement of \autoref{thm:result0}]\label{infomral_thm1}
Let $\mathcal{X}$ be the feature space and consider the model function defined in \eqref{model}. Let each parameter be initialized independently from the uniform distribution on $[0,2\pi)$, and let us assume that $\mathbb{E}_\theta\,f_\theta(x) = 0$ for any $x\in\mathcal{X}$. Then, for any $x,\,x'\in \mathcal{X}$ and any $\epsilon,\,\delta>0$ there is a classical probabilistic algorithm that estimates the number $K(x,x')$
with precision $\epsilon$ and probability at least $1-\delta$ with
\begin{equation}
    N=\frac{8L^2m^{2}}{3\epsilon^{2}}\log\frac{2}{\delta}
\end{equation}
iid samples of $\theta$, where $L$ denotes the number of parameters of the model, and where the measured observable $O$ can be expressed as a linear combination of $m$ tensor products of Pauli matrices. The algorithm requires $O(N\,L^2\,m\,n^2)$ classical operations.
\end{thm}
The key idea of the algorithm is that, since $f_\theta(x)$ is a quadratic polynomial in the entries of the unitary matrix associated to each gate and the analytic NTK is quadratic in $f$, the expectation value with respect to the uniform distribution of the parameters can be replaced by the expectation value with respect to the uniform distribution on the set $\left\{0, \frac{\pi}{2}, \pi, \frac{3\pi}{2} \right\}$.
For such values of the parameters, the parametric gates belong to the Clifford group, and since also the nonparametric gates belong to the Clifford group by hypothesis, the whole $U_{x,\theta}$ belongs to the Clifford group.
Therefore, the analytic NTK can be computed via the efficient algorithm for Clifford quantum circuits \cite{Aaronson_2004}.

Let us now state informally our second contribution.
The average of the Gaussian process associated to a trained quantum neural network has a limit $\mu_\infty:\mathcal{X}\to\mathbb{R}$ for infinite training time.
Such a limit depends only on the NTK of the network and on the training data (see \autoref{sub:NTK} and \cite{girardi2024,melchor2025quantitative} for the details).
We can now leverage the algorithm of \autoref{infomral_thm1} to estimate $\mu_\infty$:
\begin{thm}[Informal statement of \autoref{thm:result1}]\label{informal_thm2}
Within the same hypotheses of \autoref{infomral_thm1}, the following hold true:
\begin{enumerate}
\item[(a)]For any $x\in \mathcal{X}$ and any $\epsilon,\,\delta>0$ there is a classical probabilistic algorithm that, with probability at least $1-\delta$, outputs an estimate $\tilde{\mu}_\infty(x)$ of the average $\mu_\infty(x)$ of the model function in $x$ in the limit of infinite training time such that
\begin{equation}
    \left|\tilde{\mu}_\infty(x) - \mu_\infty(x)\right|<\epsilon
\end{equation}
with
\begin{align}\label{eq:mux}
N=\frac{24R^2 + 4R\epsilon}{3\epsilon^2}\log\frac{2(1+d_{{\rm train}})}{\delta} + \frac{2(1+\sqrt{2})^{4}L^4d_{{\rm train}}^3m^8\norma{K_{{\rm train}}^{-1}}_{{\rm op}}^{4}\norma{Y}_{2}^{2}}{3\,\epsilon^2}\log\frac{4d_{{\rm train}}}{\delta}
\end{align}
iid samples of $\theta$, where 
\begin{align}
&R=2L\sqrt{d_{{\rm train}}}m^{2}\norma{K_{{\rm train}}^{-1}Y}_{2},\\
\end{align}
$d_{\mathrm{train}}$ is the number of training examples, $Y$ is the vector of the training labels and $K_{{\rm train}}$ is the NTK restricted to the training inputs.
The algorithm requires
\begin{equation}\label{eq:big-O}
O\left(N\,L^{2}\,d_{{\rm train}}\left(m\,n^2+d_{{\rm train}}^{2}\right)\right)
\end{equation}
elementary operations.
\item[(b)] The output $\tilde{\mu}_\infty$ of the algorithm above satisfies
\begin{equation}
    \left|\tilde{\mu}_\infty(x) - \mu_\infty(x)\right|<\epsilon\qquad\forall\,x\in\mathcal{X}
\end{equation}
(\emph{i.e.}, the estimate of $\mu_\infty$ has a uniformly bounded error on the whole $\mathcal{X}$) with probability at least $1-\delta$ with
\begin{equation}\label{eq:mu}
N=\frac{24R^2 + 4R\epsilon}{3\epsilon^2}\log\frac{2|\mathcal{X}|(1+d_{{\rm train}})}{\delta}+\frac{2(1+\sqrt{2})^{4}L^4\,d_{{\rm train}}^3m^8\norma{K_{{\rm train}}^{-1}}_{{\rm op}}^{4}\norma{Y}_{2}^{2}}{3\,\epsilon^2}\log\frac{4|\mathcal{X}|d_{{\rm train}}}{\delta} 
\end{equation}
iid samples of $\theta$.
\end{enumerate}
\end{thm}
\begin{rem}
\eqref{eq:mu} coincides with \eqref{eq:mux} upon replacing $\delta$ with $\frac{\delta}{|\mathcal{X}|}$.
A number of samples growing logarithmically with the size of the feature space is enough to get an estimate of $\mu_\infty$ with a uniformly bounded error on the whole $\mathcal{X}$.
\end{rem}

\modifica{\begin{rem}
    \autoref{informal_thm2} states that sample and time complexities related to the classical algorithm for estimating $\mu_\infty(x)$ are polynomial. However, we do not believe that our exponents are optimal. Our result needs a time complexity scaling as $d_{\rm train}^6$ (see \eqref{eq:mux} and \eqref{eq:big-O}). For realistic values of $d_{\rm train}$ (e.g., the MNIST dataset of handwritten digits includes about $10^5$ examples), such a scaling would be prohibitive. We are confident that our upper bound to the time complexity is pessimistic and that numerical experiments, which we plan to address in a future work, may reveal a better scaling.
\end{rem}}

Our work conceptually aligns with a recent series of results \cite{angrisani2024,ReardonSmith2024improvedsimulation, PRXQuantum.6.010337, martinez25,cerezo2024doesprovableabsencebarren} demonstrating that a wide class of variational quantum circuits can be simulated efficiently by classical computers.
In particular, Ref. \cite{angrisani2024} shows that, if each layer of the circuit contains random iid one-qubit gates sampled from the Haar distribution applied to all the qubits, then, regardless of how the two-qubit gates are placed, the untrained circuit can be simulated efficiently with high probability.
\modifica{Our results imply that wide quantum neural networks with logarithmic depth trained on supervised-learning problems with the input encoded with Clifford gates can be simulated efficiently by classical computers.}

\subsection{Outline}
The organization of the present work is as follows. In \autoref{sec:preli}, we introduce the structure of the considered parametric quantum circuit, the model function and the NTK.   
In \autoref{sec:results}, \modifica{we present our algorithm to estimate the NTK and $\mu_\infty$ and state our results} about the sample and computational complexity of the algorithm. In \autoref{sec:proofs}, we prove our results. 
Finally, in \autoref{sec:concl}, we present our conclusions and outline some open questions. In \autoref{app:Bernstein}, we present Bernstein's inequality for rectangular matrices, which we will employ to prove concentration bounds for our estimate of the NTK.

\section{Preliminaries}\label{sec:preli}
Let us start by introducing the notation of the present work.
\subsection{Training data}\label{sub:trdata}
Let $\X$ be the feature space, \emph{i.e.}, the set of all the possible inputs, which we assume to be finite, and we let $\mathbb{R}$ be the output space.
\subsection{The Pauli group and the Clifford group}\label{sub:quantum}
Let $\CC^{2}$ be the Hilbert space of a single qubit. In what follows, we denote by $n\in \N$ the number of qubits of the quantum neural network. Hence, the Hilbert space of the system is $\left(\CC^{2}\right)^{\otimes n}$, with dimension $2^{n}$. Let us now briefly recall the definition of the $n$-qubit Pauli group and of the $n$-qubit Clifford group. Let
\begin{align}
{\sigma_0}\coloneqq 
\begin{bmatrix}
    1&0\\
    0&1
\end{bmatrix},
\hskip 0,2cm 
{\sigma_1}\coloneqq 
\begin{bmatrix}
    0&1\\
    1&0
\end{bmatrix},
\hskip 0,2cm 
{\sigma_2}\coloneqq 
\begin{bmatrix}
    0&-i\\
    i&0
\end{bmatrix},
\hskip 0,2cm 
{\sigma_3}\coloneqq 
\begin{bmatrix}
    1&0\\
    0&-1
\end{bmatrix}
\end{align}
be the Pauli matrices.
\begin{defn}[$n$-qubit Pauli group]
The $n$-qubit Pauli group $\mathbf{P}_n$ is made by the unitary operators acting on $\left(\mathbb{C}^2\right)^{\otimes n}$ of the form $\lambda\, \sigma_{i_1}\otimes\cdots \otimes \sigma_{i_n}$ where $\lambda\in \{\pm 1, \pm i\}$, and $i_1,\,\ldots,\,i_n\in \{0,\,\ldots,\,3\}$.
\end{defn}
\begin{defn}
The $n$-qubit Clifford group $\mathbf{C}_n$ is given by all the unitary operators $U$ acting on $\left(\mathbb{C}^2\right)^{\otimes n}$ such that
 \begin{align}
 U\mathbf{P}_n U^\dagger \subseteq \mathbf{P}_n.
 \end{align}
 That is, the Clifford group $\mathbf{C}_n$ is the normalizer of the $n$-qubit Pauli group.
\end{defn}
For details about properties of these groups, see for instance \cite{mastel2023clifford}.

\subsection{Assumptions on the architecture}\label{sub:assumpt}
For any $x\in\mathcal{X}$ and any $\theta\in[0,2\pi)^L$, we consider quantum circuits of the form
\begin{equation}\label{circuit}
U_{x,\theta} = U_x^{(L)}\,e^{-i\frac{\theta_L}{2}P_L}\,\cdots\,U_x^{(1)}\,e^{-i\frac{\theta_1}{2}P_1}\,U_x^{(0)}\,,
\end{equation}
where $U_x^{(0)},\,\ldots,\,U_x^{(L)}\in\mathbf{C}_n$ can arbitrarily depend on the input $x$, and $P_1,\,\ldots,\,P_L\in\mathbf{P}_n$ are self-adjoint. Let us clarify that the requirement of encoding the input via Clifford gates is not a severe restriction. In particular, Clifford gates allow the encoding of bit strings into the corresponding vectors of the computational basis (basis encoding) or, more generally, into any stabilizer state. 

We let the measured observable $O$ be a real linear combination of at most $m$ self-adjoint elements of the $n$-qubit Pauli group:
\begin{equation}\label{ipot1}
     O=\sum_{k=1}^{m}c_{k}\,P_{k},
    \end{equation}
where $P_{k}=P_k^\dag\in\mathbf{P}_n$ and $c_{k}\in [-1,1]$.

The model function is then
\begin{equation}\label{neuralfunct}
f_\theta(x) = \langle0^n|U_{x,\theta}^\dag\,O\,U_{x,\theta}|0^n\rangle\,,
\end{equation}
where $|0^n\rangle = |0\rangle^{\otimes n}$.

\subsection{The neural tangent kernel}\label{sub:NTK}
Given a training set
\begin{equation}
\mathcal{D}=\left\{\left(x^{(i)},y^{(i)}\right)\right\}_{i=1,\dots,d_{{\rm train}}}\subseteq \mathcal{X}\times \R\,,
\end{equation}
we will call $d_{\rm train}=|\mathcal{D}|$ the number of examples and we will represent it in a vectorized form as follows
\begin{align}
X_{{\rm train}}\coloneqq\begin{pmatrix} x^{(1)}\\x^{(2)}\\\vdots\\x^{(d_{{\rm train}})} \end{pmatrix},\qquad 
Y=\begin{pmatrix} y^{(1)}\\y^{(2)}\\\vdots\\y^{(d_{{\rm train}})} \end{pmatrix}.
\end{align}
\modifica{In what follows, we assume the following:
\begin{ass}\label{A3}
The parameter vector $\theta$ is initialized by independently sampling each entry from the uniform distribution on $[0,2\pi)$.
\end{ass}
Here, we note that the parameters are optimized via gradient descent with respect to the quadratic cost function $\L$ as follows:

\begin{align}\label{gradform1}
\frac{\de \theta_{t}}{\de\,t}=-\eta \nabla_{\theta}\L(\theta_{t}) \,,
\end{align}
where $\L$ is given by
\begin{align}
 \L(\theta)\coloneqq \frac{1}{2}\sum_{i
=1}^{d_{{\rm train}}}\left(f_{\theta}(x^{(i)})-y^{(i)}\right)^{2}\,,
\end{align}
and $\eta>0$ is the learning rate.} Given any function $g:\mathcal{X}\to\mathbb{R}$, we will often use the following notation:
\begin{align}
    g(X_{{\rm train}})\coloneqq \begin{pmatrix} g(x^{(1)})\\g(x^{(2)})\\\vdots\\g(x^{(d_{{\rm train}})}) \end{pmatrix},\qquad g(X_{{\rm train}}^T)\coloneqq \begin{pmatrix} g(x^{(1)})&g(x^{(2)})&\cdots& g(x^{(d_{{\rm train}})}) \end{pmatrix}.
\end{align}
Similarly, for any bivariate function $K:\mathcal{X}\times\mathcal{X}\to\mathbb{R}$ we will write $K(X_{{\rm train}},X_{{\rm train}}^T)$ to indicate the $d_{{\rm train}}\times d_{{\rm train}}$ matrix with entries $\left(K(X_{{\rm train}},X_{{\rm train}}^T)\right)_{ij}\coloneqq K(x^{(i)},x^{(j)})$ for $1\leq i,j\leq d_{{\rm train}}$.
We set
\begin{align}\label{fvectmod}
\begin{aligned}
F_{\theta,{\rm train}}\coloneqq
\begin{pmatrix}
f_{\theta}(x^{(1)})\\
f_{\theta}(x^{(2)})\\
\vdots\\
f_{\theta}(x^{(d_{{\rm train}})})
\end{pmatrix} 
\end{aligned}
= f_{\theta}(X_{{\rm train}})\,.
\end{align}

\begin{defn}[Empirical NTK]\label{def:ENTK}
The empirical NTK is given by the inner product between the gradients of the model function evaluated for different inputs:
\begin{align}\label{NTK1}
\hat{K}_{\theta}(x,x')\coloneqq \left(\nabla_{\theta}f_{\theta}(x)\right)^{T}\nabla_{\theta}f_{\theta}(x')\,,\qquad x,\,x'\in\mathcal{X}\,.
\end{align}
\end{defn} 

In what follows, we introduce some useful assumptions to treat the behavior of the NTK. Before, let us recall the definition of the analytic NTK.
\begin{defn}[Analytic NTK]\label{def:analk1}
The analytic NTK is the expectation of the empirical NTK with respect to the probability distribution of the parameters at initialization:
\begin{align}\label{ankd1}
 K(x,x')\coloneqq \E_\theta\,\hat{K}_{\theta}(x,x')\,,\qquad x,\,x'\in\mathcal{X}\,.  
\end{align}
\end{defn}

\begin{ass}\label{A2}
We suppose that the analytic NTK restricted to the training inputs
\begin{equation}
K_{{\rm train}}=K(X_{{\rm train}},X_{{\rm train}}^{T})
\end{equation}
is invertible. We denote with $\lambda_{\max}^{K_{{\rm train}}}$ 
and $\lambda_{\min}^{K_{{\rm train}}}$ its maximum and minimum eigenvalue, respectively.
\end{ass}
From the gradient-flow equation \eqref{gradform1} and the chain rule, one gets that
\modifica{

\begin{align}\label{gradeq1}
 \displaystyle\begin{cases}
  \frac{\de \theta_{t}}{\de t}=-\eta\nabla_{\theta}f_{\theta_{t}}(X_{{\rm train}}^{T})(F_{\theta_{t},{\rm train}}-Y),\\
 \frac{\de}{\de t}f_{\theta_{t}}(x)=-\eta\left(\nabla_{\theta}f_{\theta_{t}}(x)\right)^{T}\nabla_{\theta}f_{\theta_{t}}(X_{{\rm train}}^{T})(F_{\theta_{t},{\rm train}}-Y)\,,
 \end{cases}   
\end{align}
where $\nabla_{\theta}f_{\theta_{t}}(X_{{\rm train}}^{T})$ denotes the gradient of $f_{\theta_{t}}(X_{{\rm train}}^{T})$ with respect to $\theta$. Recall that $^T$ is the transposition operator. Let us notice that \eqref{gradeq1} can be written as
\begin{align}\label{gradeq2}
\displaystyle\begin{cases}
  \frac{\de \theta_{t}}{\de t}=-\eta\nabla_{\theta}f_{\theta_{t}}(X_{{\rm train}}^{T})(F_{\theta_{t},{\rm train}}-Y),\\
 \frac{\de}{\de t}f_{\theta_{t}}(x)=-\eta \hat{K}_{\theta_{t}}(x,X_{{\rm train}}^{T})(F_{\theta_{t},{\rm train}}-Y).
 \end{cases}     
\end{align}
If we assume that the NTK does not change during training, i.e., that $\hat{K}_{\theta_{t}} = \hat{K}_{\theta_{0}}$, then the time evolution of the model function becomes linear:
\begin{align}
\frac{\de }{\de t}f_{\theta_{t}^{\lin}}^{\lin}(x)=-\eta\hat{K}_{\theta_{0}}(x,X_{{\rm train}}^{T})e^{-\eta\hat{K}_{\theta_{0}}t}(F_{\theta_{0},{\rm train}}-Y)
\end{align}
and has the analytic solution 
\begin{align}\label{linevolmod1}
f_{\theta_{t}^{\lin}}^{\lin}(x)= f_{\theta_{0}}(x)-\hat{K}_{\theta_{0}}(x,X_{{\rm train}}^{T})\hat{K}_{\theta_{0}}^{-1}\left(\mathbbm{1}-e^{-\eta\hat{K}_{\theta_{0}}t}\right)(F_{\theta_{0},{\rm train}}-Y).
\end{align}
Ref. \cite{girardi2024} proves that this assumption holds: in the limit of infinite width with the depth growing at most logarithmically with the width, if covariance of $f$ at initialization and the analytic NTK have finite limits $\mathcal{K}_0$ and $K$, respectively, then the empirical NTK stays close to its value at initialization throughout the whole training.
Moreover, the value of the NTK at initialization is with high probability close to its average with respect to the random initialization of the parameters \cite[Theorem 4.7]{girardi2024}. As a consequence, the probability distribution of the trained model function $\{f_{\theta_{t}}(x)\}_{x\in\X}$ converges in distribution to the Gaussian process with mean and covariance given by \cite[Corollary 4.9]{girardi2024}:}
\begin{align}\label{covlimit}
&\begin{aligned}
&\K_{t}(x,x')\coloneqq \K_{0}(x,x')-K(x,X_{{\rm train}}^{T})K_{{\rm train}}^{-1}\left(\mathbbm{1}-e^{-t\eta\,K_{{\rm train}}}\right)\mathcal{K}_{0}(X_{{\rm train}},x')\\
&-K(x',X_{{\rm train}}^{T})K_{{\rm train}}^{-1}\left(\mathbbm{1}-e^{-t\eta\,K_{{\rm train}}}\right)\mathcal{K}_{0}(X_{{\rm train}},x)+\\
&+K(x,X_{{\rm train}}^{T})K_{{\rm train}}^{-1}\left(\mathbbm{1}-e^{-t\eta\,K_{{\rm train}}}\right)\mathcal{K}_{0}(X_{{\rm train}},X_{{\rm train}}^{T})\left(\mathbbm{1}-e^{-t\eta\,K_{{\rm train}}}\right)K_{{\rm train}}^{-1}K(X_{{\rm train}},x'),
\end{aligned}\\\label{newmedia}
&\mu_{t}(x)\coloneqq K(x,X_{{\rm train}}^{T})K_{{\rm train}}^{-1}\left(\mathbbm{1}-e^{-t\eta\,K_{{\rm train}}}\right)Y\,.
\end{align}
We stress that both $\K_{t}(x,x')$ and $\mu_{t}(x)$ have a limit for $t\to\infty$.
In particular, $\mu_{t}(x)$ converges to
\begin{equation}\label{eq:muinfty}
    \mu_{\infty}(x) = K(x,X_{{\rm train}}^{T})K_{{\rm train}}^{-1}\,Y\,.
\end{equation}

Ref. \cite{melchor2025quantitative} provides a quantitative version of the convergence that holds at finite width. The results of Ref. \cite{melchor2025quantitative} hold uniformly with respect to the training time, and therefore prove that the limits of infinite width and infinite training time commute, and therefore the average of the model function trained for infinite time actually converges to \eqref{eq:muinfty} in the limit of infinite width.

\section{Our Results}\label{sec:results}

\subsection{An efficient classical algorithm to estimate the analytic NTK}
In the present work, we propose an efficient classical algorithm to estimate the analytic NTK as well as $\mu_\infty$ to within any specified precision $\epsilon > 0$ with high probability. To that end, we introduce an estimator for the analytic NTK. Let $\mathcal{S} \coloneqq \{\theta^{(1)},\ldots,\theta^{(N)}\}$ be a collection of parameter samples, where each $\theta^{(i)}$ is drawn uniformly from the set $\{0,\, \frac{\pi}{2},\, \pi,\, \frac{3\pi}{2} \}^L$. Although in general the parameters of the circuit that we consider are generic angles in $[0, 2\pi)$, we prove below that due to the structure of the parametric unitary $U_{\theta, x}$, it is sufficient to restrict attention to $\{0,\, \frac{\pi}{2},\, \pi,\, \frac{3\pi}{2} \}$. 

The following lemma will allow us to compute the gradient with respect to the parameters in the definition of the NTK as an exact finite difference:
\begin{lem}[Parameter-shift rule]\label{lem:parrule}
Let $f_\theta(x)$ be the model function \eqref{neuralfunct} of a parametric quantum circuit as in \eqref{circuit}.
Then, for any $x\in\mathcal{X}$, any $\theta\in[0,2\pi)^L$ and any $i=1,\,\ldots,\,L$ we have
\begin{equation}
\partial_{\theta_i}f_\theta(x) = \frac{1}{2}\left(f_{\theta + \frac{\pi}{2} e_i}(x) - f_{\theta - \frac{\pi}{2} e_i}(x)\right)\,,
\end{equation}
where $e_i$ is the $i$-th vector of the canonical basis of $\mathbb{R}^L$.
\end{lem}
\begin{proof}
This is just a matter of computation, and we refer to \cite{girardi2024,banchi2021measuring}. 
\end{proof}

With this, we propose the following empirical estimator for the analytic NTK:
\begin{align}\label{approxestt}
    \tilde{K}(x,x') &= \frac{1}{N}\sum_{j=1}^N\sum_{i=1}^L\partial_{\theta_i}f_{\theta^{(j)}}(x)\,\partial_{\theta_i}f_{\theta^{(j)}}(x')\nonumber\\
    &=\frac{1}{4\,N}\sum_{j=1}^N\sum_{i=1}^L\left(f_{\theta^{(j)} + \frac{\pi}{2} e_i}(x) - f_{\theta^{(j)} - \frac{\pi}{2} e_i}(x)\right)\left(f_{\theta^{(j)} + \frac{\pi}{2} e_i}(x') - f_{\theta^{(j)} - \frac{\pi}{2} e_i}(x')\right)\,,
\end{align}
where we have used the parameter-shift rule to express the derivatives with respect to the parameters as finite differences.
From a statistical perspective, the estimator $\tilde{K}$ possesses several desirable properties that make it a reliable tool for approximating the analytic NTK. First, it is an unbiased estimator, meaning that its expected value coincides exactly with the true kernel (see \autoref{prop:unbiased} below). This ensures that, on average, the estimator does not systematically overestimate or underestimate the quantity of interest \cite{barankin1949locally,halmos1946theory}. Moreover, because the estimator is constructed as an average over $N$ independent and identically distributed random samples, classical concentration inequalities (such as Hoeffding’s or Bernstein’s inequality) can be applied to provide probabilistic guarantees on the deviation of $\tilde{K}$ from $K$ \cite{Tropp_2011,talagrand1995missing}.
We then propose the following classical algorithm to estimate the NTK:
\begin{algorithm}[Classical estimation of $K(x,x')$]\label{alg:1}
~\\
\begin{itemize}
\item Sample $N$ sets of parameters $\mathcal{S} = \left\{\theta^{(1)},\,\ldots,\,\theta^{(N)}\right\}$, where each entry of each $\theta^{(i)}$ is sampled from the uniform distribution on $\left\{0,\,\frac{\pi}{2},\,\pi,\,\frac{3\pi}{2}\right\}$.
\item Compute $f_{\theta\pm\frac{\pi}{2}e_i}(x)$, $f_{\theta\pm\frac{\pi}{2}e_i}(x')$ for any $\theta\in\mathcal{S}$ and any $i=1,\,\ldots,\,L$.
\item Return
\begin{align*}
    \tilde{K}(x,x') &= \frac{1}{4\,N}\sum_{j=1}^N\sum_{i=1}^L\left(f_{\theta^{(j)} + \frac{\pi}{2} e_i}(x) - f_{\theta^{(j)} - \frac{\pi}{2} e_i}(x)\right)\left(f_{\theta^{(j)} + \frac{\pi}{2} e_i}(x') - f_{\theta^{(j)} - \frac{\pi}{2} e_i}(x')\right)\,.
\end{align*}
\end{itemize}
\end{algorithm}

\autoref{alg:1} can be employed to compute a classical estimate of the average of the trained model function:
\begin{algorithm}[Classical estimation of $\mu_\infty(x)$]\label{alg:2}
~\\
\begin{itemize}
\item Sample $N$ sets of parameters $\mathcal{S} = \left\{\theta^{(1)},\,\ldots,\,\theta^{(N)}\right\}$, where each entry of each $\theta^{(i)}$ is sampled from the uniform distribution on $\left\{0,\,\frac{\pi}{2},\,\pi,\,\frac{3\pi}{2}\right\}$.
\item Apply \autoref{alg:1} with the above $\mathcal{S}$ to compute an estimate $\tilde{K}_{\mathrm{train}}$ of the NTK on the training inputs and an estimate $\tilde{K}(x,X_{\mathrm{train}}^T)$ of $K(x,X_{\mathrm{train}}^T)$.
\item Compute the inverse matrix $\tilde{K}^{-1}_{\mathrm{train}}$.
\item Return
\begin{align*}
   \tilde{\mu}_\infty(x) = \tilde{K}(x,X_{\mathrm{train}}^T)\,\tilde{K}^{-1}_{\mathrm{train}}\,Y\,.
\end{align*}
\end{itemize}
\end{algorithm}

\subsection{Sample complexity}
The following \autoref{thm:result0} determines the sample complexity of estimating the NTK via \autoref{alg:1}:
\begin{thm}[Formal statement of \autoref{infomral_thm1}]\label{thm:result0}
Let $\mathcal{X}$ be the feature space, and for each $x\in \mathcal{X}$ let us consider a parametric quantum circuit as defined in \eqref{circuit} with model function as in \eqref{neuralfunct}.
Suppose that \autoref{A3} holds true.
Let us assume that $\mathbb{E}_\theta\,f_\theta(x) = 0$ for any $x\in\mathcal{X}$.
Then for any $x,x'\in \mathcal{X}$ and any $\epsilon,\,\delta>0$, the output $\tilde{K}(x,x')$ of \autoref{alg:1} with
\begin{equation}
    N=\frac{8L^2m^{2}}{3\epsilon^{2}}\log\frac{2}{\delta}
\end{equation}
iid samples of $\theta$ satisfies 
\begin{align}
\left|\tilde{K}(x,x') - K(x,x')\right| < \epsilon
\end{align}
with probability at least $1-\delta$.
The algorithm requires $O(N\,L^{2}\,m\,n^2)$ elementary operations.
\end{thm}


The following \autoref{thm:result1} determines the sample and computational complexity of estimating $\mu_\infty$ via \autoref{alg:2}:

\begin{thm}[Formal statement of \autoref{informal_thm2}]\label{thm:result1}
Under the same hypotheses of \autoref{thm:result0}, the following hold true.
\begin{enumerate}
\item[(a)] For any $x\in\mathcal{X}$, any
\begin{equation}\label{hyponew}
0<\epsilon< \frac{L}{2}\sqrt{d_{{\rm train}}}m^2\norma{Y}_{2}\norma{K_{{\rm train}}^{-1}}_{{\rm op}}
    \end{equation}
and any $\delta>0$, the output $\tilde{\mu}_\infty(x)$ of \autoref{alg:2} with
\begin{equation}\label{eq:samplesmux}
N=\frac{24R^2 + 4R\epsilon}{3\epsilon^2}\log\frac{2(1+d_{{\rm train}})}{\delta}+\frac{2(1+\sqrt{2})^{4}L^4\,d_{{\rm train}}^3m^8\norma{K_{{\rm train}}^{-1}}_{{\rm op}}^{4}\norma{Y}_{2}^{2}}{3\,\epsilon^2}\log\frac{4d_{{\rm train}}}{\delta}    
\end{equation}
iid samples of $\theta$, where 
\begin{align}
&R=2L\sqrt{d_{{\rm train}}}m^{2}\norma{K_{{\rm train}}^{-1}Y}_{2}\,,
\end{align}
satisfies
\begin{equation}
    \left|\tilde{\mu}_\infty(x) - \mu_\infty(x)\right| < \epsilon
\end{equation}
with probability at least $1-\delta$.
The algorithm requires
\begin{align}
O(NL^{2}\,d_{{\rm train}}\,[m\,n^2+d_{{\rm train}}^{2}])
\end{align}
elementary operations.
\item[(b)] For any $\epsilon$ satisfying \eqref{hyponew} and any $\delta>0$, the output $\tilde{\mu}_\infty$ of \autoref{alg:2} applied to each $x\in\mathcal{X}$ with the same $\mathcal{S}$ with
\begin{equation}
N=\frac{24R^2 + 4R\epsilon}{3\epsilon^2}\log\frac{2|\mathcal{X}|(1+d_{{\rm train}})}{\delta}+\frac{2(1+\sqrt{2})^{4}L^4\,d_{{\rm train}}^3m^8\norma{K_{{\rm train}}^{-1}}_{{\rm op}}^{4}\norma{Y}_{2}^{2}}{3\,\epsilon^2}\log\frac{4|\mathcal{X}|d_{{\rm train}}}{\delta}    
\end{equation}
iid samples of $\theta$ satisfies
\begin{equation}
    \left|\tilde{\mu}_\infty(x) - \mu_\infty(x)\right| < \epsilon \qquad\forall\,x\in\mathcal{X}
\end{equation}
with probability at least $1-\delta$.
\end{enumerate}
\end{thm}

\begin{rem}
From part (b) of \autoref{thm:result1}, \autoref{alg:2} can estimate $\mu_\infty$ with a uniformly bounded error on the whole $\mathcal{X}$ with a number of samples that grows logarithmically with the size of the feature space.
\end{rem}
\modifica{
\begin{rem}
We stress that the algorithm itself does not require any assumptions on the width or depth of the network. Such assumptions are only needed to guarantee that $\mu_\infty$ is close to the mean of the trained model function.
\end{rem}
}
\section{Proofs}\label{sec:proofs}
In this section, we provide the proofs of the main results.

\subsection{Preliminary results}

We now prove that $\tilde{K}$ is an unbiased estimator.

\begin{prop}\label{prop:unbiased}
Suppose that \autoref{A3} holds true. For any $x,\,x'\in\mathcal{X}$, we have
\begin{align}
 \E_{\theta^{(1)},\ldots,\theta^{(N)}\sim \mathrm{unif}\left\{0,\,\frac{\pi}{2},\,\pi,\,\frac{3\pi}{2}\right\}^{L}}\,\tilde{K}(x,x') = K(x,x')\,.   
\end{align}
\end{prop}

\begin{proof}
Since each $\theta^{(i)}\in\mathcal{S}$ is independent and identically distributed, then we need to prove that for any $x,\,x'\in\mathcal{X}$ we have
\begin{align}\label{toprove}
\E_{\theta \sim [0,2\pi)^{L}}\left[\nabla_\theta f(\theta,x)^T\nabla_\theta f(\theta,x')\right]=\E_{\theta \sim \left\{0,\,\frac{\pi}{2},\,\pi,\,\frac{3\pi}{2}\right\}^{L}}\left[\nabla_\theta f(\theta,x)^T\nabla_\theta f(\theta,x')\right]\,.
\end{align}
In order to prove \eqref{toprove}, we need the following:
\begin{prop}\label{prop1}
Let $d\in \N$, and $P\in\mathbb{C}^{d\times d}$ be hermitian and such that $P^2 = I$.
Let $\mu$ be a probability distribution on $\mathbb{R}$.
Then, $\mathbb{E}_{\theta\sim\mu}\left(e^{i\frac{\theta}{2}P} \otimes e^{-i\frac{\theta}{2}P}\right)^{\otimes 2}$ depends on $\mu$ only through the quantities $\mathbb{E}_{\theta\sim\mu}\,e^{i\theta},\,\mathbb{E}_{\theta\sim\mu}\,e^{2i\theta}$.
In particular, the uniform distribution on $[0,2\pi)$ and the uniform distribution on $\left\{0,\,\frac{\pi}{2},\,\pi,\,\frac{3\pi}{2}\right\}$ result in the same expectation value of $\left(e^{i\frac{\theta}{2}P} \otimes e^{-i\frac{\theta}{2}P}\right)^{\otimes 2}$. 
\end{prop}

\begin{proof}
Let
\begin{equation}
H = \frac{P\otimes I - I \otimes P}{2}\,.
\end{equation}
We have
\begin{equation}
\left(e^{i\frac{\theta}{2}P} \otimes e^{-i\frac{\theta}{2}P}\right)^{\otimes 2} = e^{i\theta\left(H_1 + H_2\right)}\,,
\end{equation}
where each $H_i$ acts on the $i$-th subsystem.
$P$ has spectrum $\{\pm1\}$, therefore $H$ has spectrum $\{-1,\,0,\,1\}$ and $H_1 + H_2$ has spectrum $\{-2,\,\ldots,\,1,\,2\}$.
For any $l = -2,\,\ldots,\,2$, let $\Pi_l$ be the orthogonal projector onto the eigenspace of $H_1 + H_2$ with eigenvalue $l$.
Then,
\begin{equation}
    \mathbb{E}_{\theta\sim\mu}\left(e^{i\frac{\theta}{2}P} \otimes e^{-i\frac{\theta}{2}P}\right)^{\otimes 2} = \sum_{l=-2}^2 \left(\mathbb{E}_{\theta\sim\mu}\,e^{il\theta}\right)\Pi_l\,.
\end{equation}
The claim follows since for any $l=1,\,2$
\begin{equation}
\mathbb{E}_{\theta\sim\mu}\,e^{-il\theta} = \left(\mathbb{E}_{\theta\sim\mu}\,e^{il\theta}\right)^*\,.
\end{equation}
\end{proof}
Let us notice that due to \autoref{lem:parrule}, we have that

\begin{equation}\label{reducedexp}
\hat{K}_{\theta}(x,x')= \frac{1}{4}\sum_{j=1}^L\left(f_{\theta + \frac{\pi}{2} e_j}(x) - f_{\theta - \frac{\pi}{2} e_j}(x)\right)\left(f_{\theta + \frac{\pi}{2} e_j}(x') - f_{\theta - \frac{\pi}{2} e_j}(x')\right).    
\end{equation}
Denote by $u_{j}=e^{-i\frac{\theta_{j}}{2}P_{j}}$, for $j=1,\ldots, L$. Notice that $\hat{K}_{\theta}(x,x')$ can be written as a polynomial of degree $2$ in the coefficients of $u_{j}$ and degree $2$ in their complex conjugates. Hence, we have that $\hat{K}_{\theta}(x,x')$ can be written as

\begin{align}\label{ultim}
\hat{K}_{\theta}(x,x')={\rm tr}\left(\left(e^{-i\frac{\theta_{j}}{2}P_{j}}\otimes e^{i\frac{\theta_{j}}{2}P_{j}}\right)^{\otimes 2}J_j(\theta,x,x')\right)    
\end{align}
where $J_j(\theta,x,x')$ is a suitable linear operator that does not depend on $\theta_j$. Then by \autoref{prop1}, \eqref{ultim} has the same expected value with respect to the uniform distribution on $[0,2\pi)^L$ and with respect to the uniform distribution on $\left\{0,\,\frac{\pi}{2},\,\pi,\,\frac{3\pi}{2}\right\}^L$, from which we conclude that \eqref{toprove} holds true. Since all the variables in $\mathcal{S}$ are iid, the conclusion of \autoref{prop:unbiased} is proved.
\end{proof}
In what follows, we estimate the deviation of $\tilde{K}_{\mathrm{train}}$ from $K_{\mathrm{train}}$. In the next, let us denote by $\norma{\cdot}_{{\rm op}}$ the operator norm.
\begin{thm}\label{thm-1}
    Let $d_{{\rm train}}$ be the number of training examples.
    We have for any $0<t\le\left\|f\right\|_\infty^2$
    \begin{equation}
        \mathbb{P}\left\{\left\|\tilde{K}_{\mathrm{train}} - K_{\mathrm{train}}\right\|_{{\rm op}} \ge t \right\} \le 2\,d_{{\rm train}}\exp\left(- \frac{3\,N\,t^2}{8\,L^2\,d_{{\rm train}}^2\left\|f\right\|_\infty^4}\right)\,.
    \end{equation}
\end{thm}

\begin{proof}
Let us choose in \autoref{cor:Berns}
\begin{equation}
\mathbf{X}_k = \frac{1}{4}\sum_{i=1}^L\left(F_{\theta^{(k)} + \frac{\pi}{2} e_i,\,\mathrm{train}} - F_{\theta^{(k)} - \frac{\pi}{2} e_i,\,\mathrm{train}}\right)\left(F_{\theta^{(k)} + \frac{\pi}{2} e_i,\,\mathrm{train}} - F_{\theta^{(k)} - \frac{\pi}{2} e_i,\,\mathrm{train}}\right)^T - K_{\mathrm{train}}\,.
\end{equation}
We have
\begin{align}
    \mathbf{X}_k &\le \frac{1}{4}\sum_{i=1}^L\left(F_{\theta^{(k)} + \frac{\pi}{2} e_i,\,\mathrm{train}} - F_{\theta^{(k)} - \frac{\pi}{2} e_i,\,\mathrm{train}}\right)\left(F_{\theta^{(k)} + \frac{\pi}{2} e_i,\,\mathrm{train}} - F_{\theta^{(k)} - \frac{\pi}{2} e_i,\,\mathrm{train}}\right)^T\nonumber\\
    &\le \frac{1}{4}\sum_{i=1}^L\left\|F_{\theta^{(k)} + \frac{\pi}{2} e_i,\,\mathrm{train}} - F_{\theta^{(k)} - \frac{\pi}{2} e_i,\,\mathrm{train}}\right\|^2 \le L\,d_{{\rm train}}\left\|f\right\|_\infty^2.
\end{align}
We also have
\begin{equation}
\mathbf{X}_k \ge - K_{\mathrm{train}} \ge -L\,d_{{\rm train}}\left\|f\right\|_\infty^2,
\end{equation}
such that
\begin{equation}
 \left\|\mathbf{X}_k\right\|\le L\,d_{{\rm train}}\left\|f\right\|_\infty^2.
\end{equation}
We get from \autoref{cor:Berns}
\begin{equation}
        \mathbb{P}\left\{\left\|\tilde{K}_{\mathrm{train}} - K_{\mathrm{train}}\right\|_{{\rm op}} \ge t \right\} \le 2\,d_{{\rm train}}\exp\left(- \frac{N\,t^2 / 2}{L^2\,d_{{\rm train}}^2\left\|f\right\|_\infty^4 + L\,d_{{\rm train}}\left\|f\right\|_\infty^2 t / 3}\right)
    \end{equation}
Since $t\le\left\|f\right\|_\infty^2$, we have
\begin{equation}
    L^2\,d_{{\rm train}}^2\left\|f\right\|_\infty^4 + L\,d_{{\rm train}}\left\|f\right\|_\infty^2 t / 3 \le \frac{4}{3}\,L^2\,d_{{\rm train}}^2\left\|f\right\|_\infty^4\,.
\end{equation}
The claim follows.
\end{proof}

\subsection{Proof of \autoref{thm:result0}}
\subsubsection{Estimation of the number of samples}
Let $\epsilon>0$. Notice that by \autoref{thm-1} one gets
\begin{equation}
        \mathbb{P}\left\{\left\vert\tilde{K}(x,x') - K(x,x')\right\vert \ge \epsilon \right\} \le 2\exp\left(- \frac{3\,N\,\epsilon^2}{8\,L^2\,\left\|f\right\|_\infty^4}\right)\,.
    \end{equation}
 Hence, by letting 

 \begin{align}
 2\exp\left(- \frac{3\,N\,\epsilon^2}{8\,L^2\,\left\|f\right\|_\infty^4}\right)\leq \delta,    
 \end{align}
one finds that

\begin{align}
N\geq \frac{8L^2\norma{f}_{\infty}^{2}}{3\epsilon^{2}}\log\left(\frac{2}{\delta}\right),
\end{align}
and where $\norma{f}_{\infty}\leq m$.
\subsubsection{Estimation of the number of elementary operations}\label{complexitycost}
Let us now estimate the required number of operations to compute $\tilde{K}(x,x')$. To this aim, we need the following:

\begin{lem}\label{cliffordlem}
Let $P\in\mathbf{P}_n$ be self-adjoint, and let $Q\in\mathbf{P}_n$.
Then,
\begin{equation}
e^{i\frac{\theta}{2}P}\,Q\,e^{-i\frac{\theta}{2}P} = \left\{\begin{array}{ll}
   Q  & \left[P,\,Q\right]=0 \\
   \cos\theta\,Q + i\sin\theta\,P\,Q  & \left\{P,\,Q\right\} = 0
\end{array}\right.\,.
\end{equation}
In particular, $e^{-i\frac{\theta}{2}P}\in\mathbf{C}_n$ for any $\theta\in\left\{0,\,\frac{\pi}{2},\,\pi,\,\frac{3\pi}{2}\right\}$.
\end{lem}
\begin{proof}
If $P$ and $Q$ commute we have
\begin{equation}
e^{i\frac{\theta}{2}P}\,Q\,e^{-i\frac{\theta}{2}P} = e^{i\frac{\theta}{2}P}\,e^{-i\frac{\theta}{2}P}\,Q = Q\in\mathbf{P}_n\,.
\end{equation}
If $P$ and $Q$ anticommute, we have
\begin{equation}
e^{i\frac{\theta}{2}P}\,Q\,e^{-i\frac{\theta}{2}P} = e^{i\frac{\theta}{2}P}\,e^{i\frac{\theta}{2}P}\,Q = e^{i\theta P}\,Q = \left(\cos\theta\,I + i\sin\theta\,P\right)Q = \cos\theta\,Q + i\sin\theta\,P\,Q\,.
\end{equation}
For $\theta = 0,\,\frac{\pi}{2},\,\pi,\,\frac{3\pi}{2}$, respectively, we have
\begin{equation}
e^{i\frac{\theta}{2}P}\,Q\,e^{-i\frac{\theta}{2}P} = Q\,,\;i\,P\,Q\,,\;-Q\,,\;- i\,P\,Q\;\in\mathbf{P}_n\,,
\end{equation}
where we have used that $\mathbf{P}_n$ is closed with respect to multiplication.
\end{proof}

\begin{lem}\label{lem:complex1}
For any $\theta\in\left\{0,\,\frac{\pi}{2},\,\pi,\,\frac{3\pi}{2}\right\}^L$ and any $x\in\mathcal{X}$, the complexity of computing $f_{\theta}(x)$ is $O(L\,m\,n^2)$.
\end{lem}
\begin{proof}
In what follows, we need to introduce some further notation. Let us denote by $\mathbb{F}_{2}\coloneqq \{0,1\}$ the field with two elements. Let $a,b\in \mathbb{F}_{2}^{2n}$. Let
\begin{align}
\sigma_{00}=\tau_{00}={\sigma_0},
\hskip 0,2cm 
\sigma_{01}=\tau_{01}={\sigma_1},
\hskip 0,2cm 
\sigma_{10}=\tau_{10}={\sigma_3},
\hskip 0,2cm 
\sigma_{11}={\sigma_2},
\hskip 0,2cm 
\tau_{11}=i\sigma_{11}= 
\begin{bmatrix}
    0&1\\
    -1&0
\end{bmatrix}.
\end{align}
We use the following notation to indicate tensor products of Pauli matrices. For $a=\begin{pmatrix}
v\\
w
\end{pmatrix}$ where $v,w\in \mathbb{F}_{2}^{n}$, we set
\begin{align}\label{tensorproducts}
\begin{aligned}
&\sigma_{a}\coloneqq \sigma_{v_{1}}\sigma_{w_{1}}\otimes\cdots\otimes \sigma_{v_{n}}\sigma_{w_{n}},\\
&\tau_{a}\coloneqq \tau_{v_{1}}\tau_{w_{1}}\otimes\cdots \otimes\tau_{v_{n}}\tau_{w_{n}},
\end{aligned}
\end{align}
where $\sigma_{v_{i}},\sigma_{w_{i}}\in {\bf P}_{1}$.
\begin{rem}
An arbitrary element of $\mathbf{P}_n$ can be represented in a unique way by $i^\delta(-1)^\epsilon \tau_a$ where $\delta,\epsilon\in\mathbb F_2$ and $a\in\mathbb F_2^{2n}$.
In the following, we will employ such a parameterization of the Pauli group.
\end{rem}
Multiplication of two Pauli group elements can now be translated into binary terms in the following way:
\begin{lem}[{\cite[Lemma 1]{Dehaene_2003}}]\label{lem:auxdehaene}
If $a_{1},a_{2}\in \mathbb{F}_{2}^{2n}$, $\delta_{1},\delta_{2},\epsilon_{1},\epsilon_{2}\in \mathbb{F}_{2}$, and $\tau$ is defined as in \eqref{tensorproducts}, then 

\begin{align}
i^{\delta_{1}}(-1)^{\epsilon_{1}}\tau_{a_{1}}i^{\delta_{2}}(-1)^{\epsilon_{2}}\tau_{a_{2}}=i^{\delta_{12}}(-1)^{\epsilon_{12}}\tau_{a_{12}}
\end{align}
where 

\begin{align}
\begin{aligned}
&\delta_{12}=\delta_{1}+\delta_{2},\\
&\epsilon_{12}=\epsilon_{1}+\epsilon_{2}+\delta_{1}\delta_{2}+a_{2}^{T}Ua_{1},\\
&a_{12}=a_{1}+a_{2},\\
&U=
\begin{bmatrix}
0_{n}& I_{n}\\
0_{n}&0_{n}
\end{bmatrix},
\end{aligned} 
\end{align}
where the multiplication, and addition of  binary variables is modulo $2$.
\end{lem}
Let us notice that by definition a Clifford operation $Q$, maps the Pauli group ${\bf P}_{n}$ to itself under conjugation:

\begin{align}\label{hermitiannot}
 Q\tau_{a}Q^{\dagger}=i^{\delta}(-1)^{\epsilon}\tau_{b}   
\end{align}
for some parameters $\delta,\varepsilon,b$ functions of $a$. Since 
\begin{align}
Q\tau_{a_{1}}\tau_{a_{2}}Q^{\dagger}= \left(Q\tau_{a_{1}}Q^{\dagger}\right)\left(Q\tau_{a_{2}}Q^{\dagger}\right),
\end{align}
then we only need to know the image of a generating set of the Pauli group to know the image of all Pauli group elements and define $Q$ up to an overall phase. Hence, it is sufficient to know the image of $\tau_{b_{k}}$, where $\{b_{k}\}_{k=1}^{n}$ form a basis of $\mathbb{F}_{2}^{2n}$. Therefore, it is sufficient to work with Hermitian Pauli group elements only since the image of a Hermitian matrix $H$ under the map $H\mapsto QHQ^{\dagger}$ will be a Hermitian matrix. Hence, in the notation of the right-hand side of \eqref{hermitiannot} Hermitian Pauli group elements are described as 

\begin{align}
    i^{a^{T}Ua}(-1)^{\epsilon}\tau_{a}.
\end{align}
Let us notice that $a^{T}Ua$ modulo $2$ counts the number of $\tau_{11}$ (which is the unique non Hermitian matrix of the four $\tau$ matrices) in the tensor product of $\tau_{a}$. Let us denote by $(e_{k})_{k=1}^{2n}$ the standard basis of $\mathbb{F}_{2}^{2n}$ where $e_{k}$ is the $k$-th column of $I_{2n}$, and consider the hermitian generators of the Pauli group $\left\{\tau_{e_{k}}\right\}_{k=1}^{2n}$. These correspond to single qubit-operations $\sigma_{z}$, and $\sigma_{x}$. So that, we denote their images under $H\mapsto QHQ^{\dagger}$ by $i^{d_{k}}(-1)^{h_{k}}\tau_{c_{k}}$. We then assemble the matrix $C$ with columns $c_{k}$, and the vectors $d$, $h$ with scalars $d_{k}$, $h_{k}$, respectively. Since the images are Hermitian, then 

\begin{align}
d_{k}=c_{k}^{T}Uc_{k} \hskip 0,1cm \text{or}\hskip 0,1cm d={\rm diag}(C^{T}UC)  
\end{align}
where ${\rm diag}(C^{T}UC)$ denotes the vector of diagonal elements of $C^{T}UC$. Hence, given $C,d$, and $h$ defining the Clifford operation $Q$, the image $i^{\delta_{2}}(-1)^{\epsilon_{2}}\tau_{b_{2}}$ of $i^{\delta_{1}}(-1)^{\epsilon_{1}}\tau_{b_{1}}$ under $H\mapsto QHQ^{\dagger}$ can be found by multiplying those operators $i^{d_{k}}(-1)^{h_{k}}\tau_{c_{k}}$ for which $b_{1k}=1$. Then by \autoref{lem:auxdehaene}

\begin{align}
\begin{aligned}\label{represent}
&b_{2}=Cb_{1},\\ 
&\delta_{2}=\delta_{1}+d^{T}b_{1}\\
&\epsilon_{2}=\epsilon_{1}+ h^{T}b_{1}+b_{1}^{T}\left({\rm lows}(C^{T}UC+dd^{T})\right)b_{1}+\delta_{1}d^{T}b_{1},
\end{aligned}
\end{align}
where ${\rm lows}(P)$ denotes the strictly lower triangular part of a matrix $P$. Le us now consider

\begin{align}\label{images}
P_{k}\longmapsto U_{x}^{(\ell)}P_{k}U_{x}^{(\ell)\dagger}\qquad\ell=1,\ldots, L,   
\end{align}
where $P_{k}\in {\bf P}_{n}$, $k=1,\ldots, m$. Hence the images of $P_{k}$ under \eqref{images} can be written as in \eqref{represent}. Furthermore, by the explicit formulas in \eqref{represent}, we have that the number of operations to specify the image is then $O(n^{2})$ (see also \cite[Section III]{Aaronson_2004}). Then for each $\ell=1,\ldots, L$ one gets that by \autoref{cliffordlem}, 

\begin{align}
  e^{i\frac{\theta_\ell}{2}P_\ell}U_{x}^{(\ell)}P_{k}U_{x}^{(\ell)\dagger} e^{-i\frac{\theta_\ell}{2}P_\ell} 
\end{align}
belongs to the Clifford group for any $\theta_{\ell} \in \left\{0,\,\frac{\pi}{2},\,\pi,\,\frac{3\pi}{2}\right\}$, and applying \autoref{lem:auxdehaene}, its images can be computed as in \eqref{represent} by making $O(n^{2})$ elementary operations. Recall that $f_{\theta}(x)$ can be written as
\begin{equation}
f_{\theta}(x)=\sum_{k=1}^{m}c_{k}\bra{0^n}U_{x,\theta}^\dag\,P_{k}\,U_{x,\theta}\ket{0^n}
\end{equation}
where
\begin{align}
    U_{x,\theta} = U_x^{(L)}\,e^{-i\frac{\theta_L}{2}P_L}\,\cdots\,U_x^{(1)}\,e^{-i\frac{\theta_1}{2}P_1}\,U_x^{(0)},
\end{align}
so that $U_{x,\theta}P_{k}U_{x,\theta}^{\dagger}$ gives us an element of the form $i^{\delta_{2}}(-1)^{\epsilon_{2}}\tau_{b_{2}}$ which can be computed as in \eqref{represent} after $2L+1$ compositions. Then applying \autoref{lem:auxdehaene}, one gets that the complexity of computing $f_{\theta}(x)$ is $O(Lmn^2)$.
\end{proof}
\begin{lem}\label{lem:complex2}
For any $\theta\in\left\{0,\,\frac{\pi}{2},\,\pi,\,\frac{3\pi}{2}\right\}^L$, and any $x\in\mathcal{X}$, the complexity of computing $\partial_{\theta_{i}}f_{\theta}(x)$ is $O(Lmn^2)$ for each $i=1,\ldots,L$.
\end{lem}
\begin{proof}
Notice that by \autoref{lem:parrule}, we have that for each $i=1,\ldots, L,$ $\partial_{\theta_{i}}f_{\theta}(x)$ is a linear combination of $f_{\theta+\frac{\pi}{2}e_{i}}(x)$, $f_{\theta-\frac{\pi}{2}e_{i}}(x)$. Then by \autoref{lem:complex1}, the complexity of computing $\partial_{\theta_{i}}f_{\theta}(x)$ is then $O(Lmn^2)$.
\end{proof}

\begin{lem}\label{lem:complex3}
For any $\theta\in\left\{0,\,\frac{\pi}{2},\,\pi,\,\frac{3\pi}{2}\right\}^L$, and any $x,x'\in\mathcal{X}$, the complexity of computing $\hat{K}_{\theta}(x,x')$, and $\tilde{K}(x,x')$ is $O(L^2mn^2)$, and $O(NL^2mn^2)$, respectively.  
\end{lem}
\begin{proof}
Let us notice that to estimate $\hat{K}_{\theta}(x,x')$ requires to perform $L$ products of the form $\partial_{\theta_i}f_{\theta}(x)\,\partial_{\theta_i}f_{\theta}(x')$. Then by \autoref{lem:complex2}, to estimate $\hat{K}(x,x')$, we need $O(L^{2}mn^2)$ operations. Therefore, the complexity of computing $\tilde{K}(x,x')$ is then $O(NL^{2}mn^2)$.
\end{proof}

\subsection{Proof of \autoref{thm:result1}}

We split the proof in two parts. In the first one, we estimate the number of samples $N$.
\subsubsection{Estimation of the number of samples}
Notice that
\begin{align}
\begin{aligned}
\mu_{\infty}(x)-\tilde{\mu}_\infty(x)&= \mu_{\infty}(x)- \tilde{K}(x,X_{{\rm train}}^{T}){K}_{\mathrm{train}}^{-1}Y +\tilde{K}(x,X_{{\rm train}}^{T}){K}_{\mathrm{train}}^{-1}Y-\tilde{\mu}(x)\\
&=\left(K(x,X_{{\rm train}}^{T})-\tilde{K}(x,X_{{\rm train}}^{T})\right)K_{\mathrm{train}}^{-1}Y +\tilde{K}(x,X_{{\rm train}}^{T})\left(K_{\mathrm{train}}^{-1}-\tilde{K}_{\mathrm{train}}^{-1}\right)Y.
\end{aligned}
\end{align}
Then for each $\epsilon>0$ such that

\begin{align}
 \epsilon\leq \vert\mu_{\infty}(x)-\tilde{\mu}(x)\vert\leq & \left\vert\left(K(x,X_{{\rm train}}^{T})-\tilde{K}(x,X_{{\rm train}}^{T})\right)K_{{\rm train}}^{-1}Y\right\vert+\\
 &+\left\vert\tilde{K}(x,X_{{\rm train}}^{T})\left(K_{{\rm train}}^{-1}-\tilde{K}_{{\rm train}}^{-1}\right)Y\right\vert
\end{align}
one gets

\begin{align}
\begin{aligned}
&\mathbb{P}\left(\epsilon\leq \vert\mu_{\infty}(x)-\tilde{\mu}(x)\vert\right)\leq \\
&\phantom{spazio}\mathbb{P}\left(\epsilon/2\leq \left\vert\left(K(x,X_{{\rm train}}^{T})-\tilde{K}(x,X_{{\rm train}}^{T})\right)K_{{\rm train}}^{-1}Y\right\vert\right)+\\
&\phantom{spazio}+\mathbb{P}\left(\epsilon/2\leq \left\vert\tilde{K}(x,X_{{\rm train}}^{T})\left(K_{{\rm train}}^{-1}-\tilde{K}_{{\rm train}}^{-1}\right)Y\right\vert\right).
\end{aligned}
\end{align}

Let us now consider the first term in the right-hand side of the last identity. Let us set

\begin{align}
\begin{aligned}
\widetilde{\mathbf{X}}_{k}^{x,{\rm train}}&\coloneqq \frac{1}{4}\sum_{i=1}^L\left(f_{\theta^{(k)} + \frac{\pi}{2} e_i}(x) - f_{\theta^{(k)} - \frac{\pi}{2} e_i}(x)\right)\left(F_{\theta^{(k)} + \frac{\pi}{2} e_i,\,\mathrm{train}} - F_{\theta^{(k)} - \frac{\pi}{2} e_i,\,\mathrm{train}}\right)^T\\
&\phantom{spaz}- K(x,X_{{\rm train}}^{T}),  
\end{aligned}
\end{align}
and also let us set

\begin{equation}
\mathbf{X}_k^{x,{\rm train}}\coloneqq \widetilde{\mathbf{X}}_{k}^{x,{\rm train}}K_{{\rm train}}^{-1}Y.
\end{equation}
Notice that 
\begin{align}
\left\vert\mathbf{X}_k^{x,{{\rm train}}}\right\vert&\leq \norma{\widetilde{\mathbf{X}}_{k}^{x,{\rm train}}}_{2}\norma{K_{{\rm train}}^{-1}Y}_{2}.
\end{align}
Notice that 
\begin{align}
\begin{aligned}
\norma{\widetilde{\mathbf{X}}_{k}^{x,{\rm train}}}_{{\rm 2}}&\leq \frac{1}{4}\sum_{i=1}^{L}\left\vert f_{\theta^{(k)} + \frac{\pi}{2} e_i}(x) - f_{\theta^{(k)} - \frac{\pi}{2} e_i}(x)\right\vert\norma{F_{\theta^{(k)}
+\frac{\pi}{2} e_i,\,\mathrm{train}} - F_{\theta^{(k)} - \frac{\pi}{2} e_i,\,\mathrm{train}}}_{2}\\ &\phantom{formu}+\norma{ K(x,X_{{\rm train}}^{T})}_{2},  
\end{aligned}
\end{align}
where
\begin{align}
\left\vert f_{\theta^{(k)} + \frac{\pi}{2} e_i}(x) - f_{\theta^{(k)} - \frac{\pi}{2} e_i}(x)\right\vert&\leq \vert f_{\theta^{(k)}_{2} + \frac{\pi}{2} e_i}(x) +\vert f_{\theta^{(k)} - \frac{\pi}{2} e_i}(x)\vert\\
&\leq 2m,
\end{align}
and similarly  
\begin{align}
\norma{F_{\theta^{(k)} + \frac{\pi}{2} e_i,\,\mathrm{train}} - F_{\theta^{(k)} - \frac{\pi}{2} e_i,\,\mathrm{train}}}_{2}\leq 2\sqrt{d_{{\rm train}}}m.
\end{align}
Hence,
\begin{align}
\norma{\widetilde{\mathbf{X}}_{k}^{x,{\rm train}}}_{2}\leq L\sqrt{d_{{\rm train}}}m^2+\norma{ K(x,X_{{\rm train}}^{T})}_{2}. 
\end{align}
Now, by applying the previous reasoning, we have $\norma{ K(x,X_{{\rm train}}^{T})}_{2}\leq L\sqrt{d_{{\rm train}}}m^2$, and thus

\begin{align}
\norma{\widetilde{\mathbf{X}}_{k}^{x,{\rm train}}}_{2}\leq 2L\sqrt{d_{{\rm train}}}m^{2}\eqqcolon \tilde{R},
\end{align}
and so that
\begin{align}
\nu(\widetilde{\mathbf{X}}^{x,{\rm train}})\leq N\tilde{R}^2,   
\end{align}
where 

\begin{align}
 \widetilde{\mathbf{X}}^{x,{\rm train}}\coloneqq \sum_{k=1}^{N}\widetilde{\mathbf{X}}_{k}^{x,{\rm train}},   
\end{align}
and $\nu(\cdot)$ as defined in \eqref{covarit}. Therefore, by \autoref{cor:Berns}
\begin{align}
\begin{aligned}
    &\mathbb{P}\left(\left\|\frac{1}{N}\sum_{k=1}^N \mathbf{X}_k^{x,{\rm train}} \right\|_{2} \ge \epsilon/2 \right) \leq \mathbb{P}\left(\norma{K_{{\rm train}}^{-1}Y}_{2}\left\|\frac{1}{N}\sum_{k=1}^N \widetilde{\mathbf{X}}_k^{x,{\rm train}} \right\|_{2} \ge \epsilon/2\right) \\
    & \phantom{spaz}= \mathbb{P}\left(\left\|\frac{1}{N}\sum_{k=1}^N \widetilde{\mathbf{X}}_k^{x,{\rm train}} \right\|_{2} \ge \norma{K_{{\rm train}}^{-1}Y}_{2}^{-1}\epsilon/2 \right)\\    
   &\phantom{spaz} \le (1+d_{{\rm train}})\cdot\exp\left(- \frac{N\,\epsilon^2 /8}{\norma{K_{{\rm train}}^{-1}Y}_{2}^{2}\tilde{R}^2 + \norma{K_{{\rm train}}^{-1}Y}_{2}\tilde{R}\,\epsilon /6}\right).
\end{aligned}
\end{align}
It remains to bound $\tilde{K}(x,X_{{\rm train}}^{T})\left(K_{{\rm train}}^{-1}-\tilde{K}_{{\rm train}}^{-1}\right)Y$. Let us define for $0\leq t\leq 1$

\begin{align}
    K_{t}\coloneqq (1-t)K_{{\rm train}}+ t\tilde{K}_{{\rm train}}.
\end{align}
Notice that 

\begin{align}
    \frac{{\rm d} K_{t}^{-1}}{{\rm d}t}=K_{t}^{-1}(\tilde{K}_{{\rm train}}-K_{{\rm train}})K_{t}^{-1}.
\end{align}
Then we have that

\begin{align}\label{derivazione}
 \frac{{\rm d}\norma{K_{t}^{-1}}_{{\rm op}}}{{\rm d}t}\leq \norma{\frac{{\rm d} K_{t}^{-1}}{{\rm d}t}}_{{\rm op}}\leq \norma{K_{t}^{-1}}_{{\rm op}}^2\norma{\tilde{K}_{{\rm train}}-K_{{\rm train}}}_{{\rm op}}.
\end{align}
To solve the above differential equation, we let 
\begin{align}
    \phi(t)=\exp\left(-\frac{1}{\norma{K_{{\rm train}}-\tilde{K}_{{\rm train}}}_{{\rm op}}y_{t}}\right), \hskip 0,2cm y_{t}=\norma{K_{t}^{-1}}_{{\rm op}}.
\end{align}
Then we obtain that

\begin{align}
    \frac{{\rm d}\phi(t)}{{\rm d}t}\leq \phi(t),
\end{align}
and so

\begin{align}
\phi(t)\leq \phi(0)\exp(t).    
\end{align}
Therefore one has that

\begin{align}\label{boundt}
y_{t}\leq \frac{y_{0}}{1-t\norma{K_{{\rm train}}-\tilde{K}_{{\rm train}}}_{{\rm op}}y_{0}}, \hskip 0,2cm y_{0}=\norma{K_{{\rm train}}^{-1}}_{{\rm op}}
\end{align}
as soon as $1-t\norma{K_{{\rm train}}-\tilde{K}_{{\rm train}}}_{{\rm op}}y_{0}>0$. In what follows, we are interested to consider $y_{t}$ with $t=1$, and the case where \eqref{boundt} holds true. Notice that by \eqref{derivazione}, for $\epsilon>0$
\begin{align}
\begin{aligned}
\epsilon/2&\leq\left\vert\tilde{K}(x,X_{{\rm train}}^{T})\left(K_{{\rm train}}^{-1}-\tilde{K}_{{\rm train}}^{-1}\right)Y\right\vert\\
&\leq \norma{\tilde{K}(x,X_{{\rm train}}^{T})}_{2}\norma{K_{{\rm train}}^{-1}-\tilde{K}_{{\rm train}}^{-1}}_{{\rm op}}\norma{Y}_{2}\\
&\leq \norma{\tilde{K}(x,X_{{\rm train}}^{T})}_{2}\norma{Y}_{2}\int_{0}^{1}\frac{{\rm d}}{{\rm dt}}\norma{K_{t}^{-1}}_{{\rm op}}{\rm dt}\\
&\leq \norma{\tilde{K}(x,X_{{\rm train}}^{T})}_{2}\norma{Y}_{2}\int_{0}^{1}\norma{\frac{{\rm d}}{{\rm dt}}K_{t}^{-1}}_{{\rm op}}{\rm dt}\\
&\leq \norma{\tilde{K}(x,X_{{\rm train}}^{T})}_{2}\norma{Y}_{2}\int_{0}^{1}\norma{K_{t}^{-1}}_{{\rm op}}^2\norma{\tilde{K}-K}_{{\rm op}}{\rm dt}\\
&\leq  \norma{\tilde{K}(x,X_{{\rm train}}^{T})}_{2}\left(\frac{\norma{K_{{\rm train}}-\tilde{K}_{{\rm train}}}_{\rm op}\norma{K_{{\rm train}}^{-1}}_{{\rm op}}^{2}}{\left(1-\norma{K_{{\rm train}}-\tilde{K}_{{\rm train}}}_{{\rm op}}\norma{K_{{\rm train}}^{-1}}_{{\rm op}}\right)^{2}}\right)\norma{Y}_{2}
\end{aligned}
\end{align}
where in the last inequality we have used \eqref{boundt}. Let us define
\begin{align}\label{term_a}
 &a\coloneqq\norma{K_{{\rm train}}-\tilde{K}_{{\rm train}}}_{{\rm op}}\norma{K_{{\rm train}}^{-1}}_{{\rm op}}.
\end{align}
We have
\begin{align}
\begin{aligned}
&\mathbb{P}\left(\epsilon/2\leq\left\vert\tilde{K}(x,X_{{\rm train}}^{T})\left(K_{{\rm train}}^{-1}-\tilde{K}_{{\rm train}}^{-1}\right)Y\right\vert\right)\leq\\
&\leq \mathbb{P}\left(\epsilon/2\leq\left\vert\tilde{K}(x,X_{{\rm train}}^{T})\left(K_{{\rm train}}^{-1}-\tilde{K}_{{\rm train}}^{-1}\right)Y\right\vert, 1-a\leq 0\right)+ \\
&\phantom{formu}+\mathbb{P}\left(\epsilon/2\leq\left\vert\tilde{K}(x,X_{{\rm train}}^{T})\left(K_{{\rm train}}^{-1}-\tilde{K}_{{\rm train}}^{-1}\right)Y\right\vert, 1-a> 0\right) \\
&\leq \mathbb{P}\left(1-a\leq 0\right)+\\
&+\mathbb{P}\left(\epsilon/2\leq \norma{\tilde{K}(x,X_{{\rm train}}^{T})}_{2}\left(\frac{\norma{K_{{\rm train}}-\tilde{K}_{{\rm train}}}_{\rm op}\norma{K_{{\rm train}}^{-1}}_{{\rm op}}^{2}}{(1-\norma{K_{{\rm train}}-\tilde{K}_{{\rm train}}}_{{\rm op}}\norma{K_{{\rm train}}^{-1}}_{{\rm op}})^{2}}\right)\norma{Y}_{2},1-a>0\right)\\
&=\mathbb{P}\left(1-a\leq 0\right)+ \mathbb{P}\left(1-a>0, \epsilon/2 \leq \frac{b}{(1-a)^2}\right),
\end{aligned}
\end{align}
where $a$ is defined in \eqref{term_a}, and $b$ is defined

\begin{align}
&b\coloneqq \norma{\tilde{K}(x,X_{{\rm train}}^{T})}_{2}\norma{K_{{\rm train}}-\tilde{K}_{{\rm train}}}_{\rm op}\norma{K_{{\rm train}}^{-1}}_{{\rm op}}^{2}\norma{Y}_{2}.
\end{align}
Notice that since $\norma{\tilde{K}(x,X_{{\rm train}}^{T})}_{2}\leq L\sqrt{d_{{\rm train}}}m^2$, then

\begin{align}
b\leq  L\sqrt{d_{{\rm train}}}m^2\norma{K_{{\rm train}}-\tilde{K}_{{\rm train}}}_{\rm op}\norma{K_{{\rm train}}^{-1}}_{{\rm op}}^{2}\norma{Y}_{2}\eqqcolon \tilde{b},
\end{align}
so that
\begin{align}
\begin{aligned}
\mathbb{P}\left(1-a\leq 0\right) &+ \mathbb{P}\left(1-a>0, \epsilon/2 \leq \frac{b}{(1-a)^2}\right)\\
&\leq \mathbb{P}\left(1-a\leq 0\right)+ \mathbb{P}\left(1-a>0, \epsilon/2 \leq \frac{\tilde{b}}{(1-a)^2}\right)\\
&=\mathbb{P}\left(1-a\leq \sqrt{\frac{\tilde{b}}{\epsilon/2}}\right).
\end{aligned}
\end{align}
Notice that 

\begin{align}
\begin{aligned}
 a+ \frac{\tilde{b}^{\frac{1}{2}}}{(\epsilon/2)^{\frac{1}{2}}}&= \norma{K_{{\rm train}}-\tilde{K}_{{\rm train}}}_{{\rm op}}\norma{K_{{\rm train}}^{-1}}_{{\rm op}}+ \\
 &+(\epsilon/2)^{-1/2}\left(L\sqrt{d_{{\rm train}}}m^2\norma{K_{{\rm train}}-\tilde{K}_{{\rm train}}}_{\rm op}\norma{K_{{\rm train}}^{-1}}_{{\rm op}}^{2}\norma{Y}_{2}\right)^{\frac{1}{2}}
\end{aligned}
\end{align}
Let us consider the equation

\begin{align}
 1=a+ \frac{\tilde{b}^{\frac{1}{2}}}{(\epsilon/2)^{\frac{1}{2}}},   
\end{align}
and we solve it for $\norma{K_{{\rm train}}-\tilde{K}_{{\rm train}}}_{{\rm op}}^{\frac{1}{2}}=x$. Then we obtain

\begin{align}
\begin{aligned}
1&=x^{2}\norma{K_{{\rm train}}^{-1}}_{{\rm op}}+x(\epsilon/2)^{-\frac{1}{2}}\sqrt{L\sqrt{d_{{\rm train}}}m^2\norma{K_{{\rm train}}^{-1}}_{{\rm op}}^{2}\norma{Y}_{2}}\\
&\eqqcolon x^2\alpha + x(\epsilon/2)^{-\frac{1}{2}}\beta
\end{aligned}
\end{align}
whose solutions for $x$ are given by

\begin{align}\label{solutions}
x=\frac{-\beta(\epsilon/2)^{-\frac{1}{2}}\pm \sqrt{(\epsilon/2)^{-1}\beta^2+4\alpha}}{2\alpha}. 
\end{align}
Since we have imposed the relation $\norma{K_{{\rm train}}-\tilde{K}_{{\rm train}}}_{{\rm op}}^{\frac{1}{2}}=x$, it implies that $x\geq 0$, and also that the negative solution in \eqref{solutions} is excluded. So that, from here the only solution that we need is 

\begin{align}
\begin{aligned}
x=\frac{-\beta(\epsilon/2)^{-\frac{1}{2}}+ \sqrt{(\epsilon/2)^{-1}\beta^2+4\alpha}}{2\alpha}&=\frac{4\alpha}{2\alpha(\beta(\epsilon/2)^{-\frac{1}{2}}+ \sqrt{(\epsilon/2)^{-1}\beta^2+4\alpha})}\\
&=\frac{2}{\beta(\epsilon/2)^{-\frac{1}{2}}+ \sqrt{(\epsilon/2)^{-1}\beta^2+4\alpha}},
\end{aligned}
\end{align}
 and where we are interested in the case where $\norma{K_{{\rm train}}-\tilde{K}_{{\rm train}}}_{{\rm op}}^{\frac{1}{2}}\geq x$. Now by \autoref{hyponew}, one has that

\begin{align}
4\alpha \leq (\epsilon/2)^{-1}\beta^{2},
\end{align}
and thus

\begin{align}
\norma{K_{{\rm train}}-\tilde{K}_{{\rm train}}}_{{\rm op}}^{\frac{1}{2}}\geq \frac{2}{\beta(\epsilon/2)^{-\frac{1}{2}}(1+\sqrt{2})}
\end{align}
that is,

\begin{align}
\norma{K_{{\rm train}}-\tilde{K}_{{\rm train}}}_{{\rm op}}\geq \frac{4}{\beta^{2}(\epsilon/2)^{-1}(1+\sqrt{2})^{2}}    
\end{align}
Therefore,

\begin{align}
\mathbb{P}\left(1-a\leq \sqrt{\frac{\tilde{b}}{\epsilon/2}}\right)\leq \mathbb{P}\left(\norma{K_{{\rm train}}-\tilde{K}_{{\rm train}}}_{{\rm op}}\geq \frac{2\epsilon}{\beta^{2}(1+\sqrt{2})^{2}}\right)   
\end{align}
where 

\begin{align}
\beta^{2}= L\sqrt{d_{{\rm train}}}m^2\norma{K_{{\rm train}}^{-1}}_{{\rm op}}^{2}\norma{Y}_{2}. 
\end{align}
Therefore, by \autoref{thm-1}

\begin{align}
& \mathbb{P}\left(\epsilon/2\leq\left\vert\tilde{K}(x,X_{{\rm train}}^{T})\left(K_{{\rm train}}^{-1}-\tilde{K}_{{\rm train}}^{-1}\right)Y\right\vert\right)\leq  \\ 
 &\phantom{formu}\leq 2\,d_{{\rm train}}\exp\left(- \frac{6\,N\,\epsilon^2}{4(1+\sqrt{2})^{4}L^4d_{{\rm train}}^3m^8\norma{K_{{\rm train}}^{-1}}_{{\rm op}}^{4}\norma{Y}_{2}^{2}}\right).
\end{align}

Therefore,

\begin{align}
\begin{aligned}
\mathbb{P}\left(\epsilon\leq\left\vert\mu_{\infty}(x)-\tilde{\mu}(x)\right\vert\right)&\leq (1+d_{{\rm train}})\cdot\exp\left(- \frac{N\,\epsilon^2 /8}{R^2 + R\,\epsilon /6}\right) +\\
&+ 2\,d_{{\rm train}}\exp\left(- \frac{6\,N\,\epsilon^2}{4(1+\sqrt{2})^{4}L^4d_{{\rm train}}^3m^8\norma{K_{{\rm train}}^{-1}}_{{\rm op}}^{4}\norma{Y}_{2}^{2}}\right),    
\end{aligned}
\end{align}
where
\begin{align}
&R=2L\sqrt{d_{{\rm train}}}m^{2}\norma{K_{{\rm train}}^{-1}Y}_{2}.
\end{align}
Let us now impose that 

\begin{align}
&(1+d_{{\rm train}})\cdot\exp\left(- \frac{N\,\epsilon^2 /8}{R^2 + R\,\epsilon /6}\right)\leq \frac{\delta}{2},\\
& 2\,d_{{\rm train}}\exp\left(- \frac{6\,N\,\epsilon^2}{4(1+\sqrt{2})^{4}L^4d_{{\rm train}}^3m^8\norma{K_{{\rm train}}^{-1}}_{{\rm op}}^{4}\norma{Y}_{2}^{2}}\right)\leq \frac{\delta}{2}.
\end{align}
By solving for $N$ we obtain the following lower bounds:

\begin{align}
&N\geq \frac{48R^2 + 8R\epsilon}{6\epsilon^2}\log\left(\frac{2(1+d_{{\rm train}})}{\delta}\right),\\
&N\geq \frac{4(1+\sqrt{2})^{4}L^4d_{{\rm train}}^3m^8\norma{K_{{\rm train}}^{-1}}_{{\rm op}}^{4}\norma{Y}_{2}^{2}}{6\,\epsilon^2}\log\left(\frac{4d_{{\rm train}}}{\delta}\right).
\end{align}
Then, we need that 

\begin{align}
\begin{aligned}
N&\geq  \frac{24R^2 + 4R\epsilon}{3\epsilon^2}\log\left(\frac{2(1+d_{{\rm train}})}{\delta}\right)\\
&+\frac{2(1+\sqrt{2})^{4}L^4d_{{\rm train}}^3m^8\norma{K_{{\rm train}}^{-1}}_{{\rm op}}^{4}\norma{Y}_{2}^{2}}{3\,\epsilon^2}\log\left(\frac{4d_{{\rm train}}}{\delta}\right).    
\end{aligned}
\end{align}

In what follows, we determine the number of iid samples of $\theta$ needed to estimate $\mu_{\infty}$ with a uniformly bounded error on the whole $\mathcal{X}$.
That is, we want to determine $N$ such that $\max_{x\in\mathcal{X}}\left|\mu_{\infty}(x)-\tilde{\mu}_\infty(x)\right|<\epsilon$.
Notice that by item (a), for any $\varepsilon>0$, and any $\delta'=\frac{\delta}{\vert \mathcal{X}\vert}$, $0<\delta<1$, we have 
\begin{align}\label{Linfty}
\begin{aligned}
\mathbb{P}\left(\epsilon\leq\norma{\mu_{\infty}-\tilde{\mu}_\infty}_{\infty}\right)&\leq\mathbb{P}\left(\bigcup_{x\in\mathcal{X}}\left\{\epsilon\leq\left\vert\mu_{\infty}(x)-\tilde{\mu}_{\infty}(x)\right\vert\right\}\right) \\
&\leq \sum_{x\in\mathcal{X}}\mathbb{P}\left(\epsilon\leq\left\vert\mu_{\infty}(x)-\tilde{\mu}_{\infty}(x)\right\vert\right)\\
&\leq \vert \mathcal{X}\vert\delta'\\
&=\delta.
\end{aligned}
\end{align}
By the previous arguments, we conclude that $\mu_{\infty}$ can be estimated with a uniformly bounded error on the whole $\mathcal{X}$ with a number of samples equal to \eqref{eq:samplesmux} with $\delta$ replaced by $\frac{\delta}{|\mathcal{X}|}$.

\subsubsection{Estimation of the number of elementary operations}
Let us now to estimate the number of elementary operations required to compute $\mu_{\infty}(x)$. Recall that by \autoref{lem:complex1} and \autoref{lem:complex2}, the number of operations required to calculate $f_{\theta}(x)$ and $\nabla_{\theta}f_{\theta}(x)$ is $O(Lmn^2)$, and $O(L^2mn^2)$ respectively. Furthermore, the complexity of computing $\nabla_{\theta}f_{\theta}(X_{{\rm train}})$ is $O(d_{{\rm train}}L^2mn^2)$. On the other hand, to compute the product $\nabla_{\theta}f_{\theta}(x)\nabla_{\theta}f_{\theta}(X_{{\rm train}}^{T})$ we need to perform $Ld_{{\rm train}}$ products. Hence, to compute $K(x,X_{{\rm train}}^{T})$, we need  $O(NL[d_{{\rm train}}Lmn^2+ d_{{\rm train}}])$ operations. On the other hand, by \cite{strassen1969gaussian}, the complexity of inverting $K_{{\rm train}}$ is $O(d_{{\rm train}}^{3})$. Hence, the complexity of computing $K_{{\rm train}}^{-1}$ is $O(NL[Lmn^2d_{{\rm train}}+d_{{\rm train}}^{2}]+ d_{{\rm train}}^{3})$. Therefore, the complexity of computing $\mu_{\infty}(x)$ is then

\begin{align}
O(NL[Lmn^2d_{{\rm train}}+d_{{\rm train}}+d_{{\rm train}}^{2}]+d_{{\rm train}}^{3} + d_{{\rm train}}).
\end{align}
In conclusion, the algorithm requires
\begin{align}
O(NL^{2}\,d_{{\rm train}}\,[m\,n^2+d_{{\rm train}}^{2}])
\end{align}
elementary operations.

\section{Conclusions}\label{sec:concl}
We have proposed an efficient classical algorithm to estimate the NTK arising from a broad family of parametric quantum circuits constructed using Clifford unitaries that depend arbitrarily on the input interleaved by parametric Pauli rotations. Despite the parametric Pauli rotations belong to the Clifford group only for few special values of the parameters, we demonstrated that the proposed \autoref{alg:1} can always efficiently estimate the NTK of the model with provable accuracy.
Moreover, we have showed that \autoref{alg:1} can be combined with the recently proved equivalence between very wide trained quantum neural networks and Gaussian processes \cite{girardi2024,melchor2025quantitative} to efficiently estimate the mean of the trained model function in the limit of infinite width (\autoref{alg:2}).
Therefore, our results prove that \modifica{wide quantum neural networks with logarithmic depth trained on supervised-learning problems with the input encoded with Clifford gates can be simulated efficiently by classical computers and cannot achieve quantum advantage.}

Our work aligns with a recent series of results \cite{ReardonSmith2024improvedsimulation, PRXQuantum.6.010337, martinez25, cerezo2024doesprovableabsencebarren,angrisani2024} showing that the output of several parametric quantum circuits can be estimated efficiently using classical resources. These findings suggest a need for a more refined understanding of where true quantum advantage lies, and under what architectural constraints an efficient classical simulation is feasible.

Our results open the way to several promising research directions. For instance, it would be interesting to investigate whether similar classical simulability results hold for broader classes of quantum circuits beyond the ansatz considered here.
In particular, it would be interesting to consider the scenarios where the input is not encoded via Clifford gates, or when at initialization the parameters are sampled from a distribution that is different from the uniform one.
Such research directions could shed light on which architectures for quantum neural networks can have hope of quantum advantages.

Finally, the proposed algorithm to efficiently compute the NTK of complex parametric quantum circuits can be viewed as a method to synthesize expressive kernels for classical learning tasks, independently of whether the kernel originates from a physically implemented quantum circuit. 
In this view, a parametric quantum circuit serves as a generative mechanism for feature maps and the associate kernel, computed by our algorithm, can be directly used in a classical kernel method. This opens the possibility of leveraging quantum structural priors without relying on quantum hardware.  


\section*{Acknowledgements}
GDP has been supported by the HPC Italian National Centre for HPC, Big Data and Quantum Computing -- Proposal code CN00000013 -- CUP J33C22001170001 and by the Italian Extended Partnership PE01 -- FAIR Future Artificial Intelligence Research -- Proposal code PE00000013 -- CUP J33C22002830006 under the MUR National Recovery and Resilience Plan funded by the European Union -- NextGenerationEU.
Funded by the European Union -- NextGenerationEU under the National Recovery and Resilience Plan (PNRR) -- Mission 4 Education and research -- Component 2 From research to business -- Investment 1.1 Notice Prin 2022 -- DD N. 104 del 2/2/2022, from title ``understanding the LEarning process of QUantum Neural networks (LeQun)'', proposal code 2022WHZ5XH -- CUP J53D23003890006.
DP has been supported by project SERICS (PE00000014) under the MUR National Recovery and Resilience Plan funded by the European Union -- NextGenerationEU.
GDP and DP are members of the ``Gruppo Nazionale per la Fisica Matematica (GNFM)'' of the ``Istituto Nazionale di Alta Matematica ``Francesco Severi'' (INdAM)''. AMH has been supported by project PRIN 2022 
``understanding the LEarning process of QUantum Neural networks (LeQun)'', proposal code 2022WHZ5XH -- CUP J53D23003890006. The author AMH is a member of the ``Gruppo Nazionale per l'Analisi Matematica, la Probabilità e le loro Applicazioni (GNAMPA)'' of the ``Istituto Nazionale di Alta Matematica ``Francesco Severi'' (INdAM)''.

\section*{Declarations}
\noindent
{\bf Conflict of interest} The authors confirm that there is no Conflict of interest.

\appendix

\section{Bernstein's upper bound for random matrices}\label{app:Bernstein}
In what follows, we aim to recall the so-called Bernstein's inequality. In the next, we denote by $\norma{\cdot}_{{\rm op}}$ the operator norm.
%
\begin{thm}[Matrix Bernstein {\cite[Theorem 6.1.1]{tropp2015}}]\label{thm:bernstein}
Consider a finite sequence $\left\{\mathbf{X}_k\right\}$ of independent, random matrices with dimension $d_{1}\times d_{2}$.
Assume that each random matrix satisfies
\begin{equation}
    \mathbb{E}\,\mathbf{X}_k = \mathbf{0} \quad \text{and} \quad \norma{\mathbf{X}_k}_{{\rm op}} \le R \quad \text{almost surely.}
\end{equation}
Introduce the matrix

\begin{equation}
\mathbf{X}\coloneqq \sum_k \mathbf{X}_k.  
\end{equation}
Let $\nu(\mathbf{X})$ be the matrix variance statistic of the sum:

\begin{equation}\label{covarit}
\nu(\mathbf{X})\coloneqq \max\left\{\norma{\sum_{k}\mathbb{E}(\mathbf{X}_{k}\mathbf{X}_{k}^{T})}_{{\rm op}},\norma{\sum_{k}\mathbb{E}(\mathbf{X}_{k}^{T}\mathbf{X}_{k})}_{{\rm op}} \right\}. \end{equation}

Then, for all $t\ge0$,
\begin{equation}
    \mathbb{P}\left\{ \norma{\mathbf{X}}_{{\rm op}}\geq t \right\} \le (d_{1}+d_{2})\cdot\exp\left(- \frac{t^2 / 2}{\nu(\mathbf{X}) + R\,t / 3}\right).
\end{equation}
\end{thm}
\begin{cor}\label{cor:Berns}
Consider a finite sequence $\left\{\mathbf{X}_1,\,\ldots,\,\mathbf{X}_N\right\}$ of independent and identically distributed random matrices with dimension $d_{1}\times d_{2}$.
Assume that each random matrix satisfies
\begin{equation}
    \mathbb{E}\,\mathbf{X}_k = \mathbf{0} \quad \text{and} \quad \left\|\mathbf{X}_k\right\|_{{\rm op}} \le R \quad \text{almost surely.}
\end{equation}
Then, for all $t\ge0$,
\begin{equation}
    \mathbb{P}\left\{\left\|\frac{1}{N}\sum_{k=1}^N \mathbf{X}_k \right\|_{{\rm op}} \ge t \right\} \le (d_{1}+d_{2})\cdot\exp\left(- \frac{N\,t^2 / 2}{\nu(\mathbf{X}) + R\,t / 3}\right).
\end{equation}
\end{cor}

\bibliographystyle{plainnat}
\bibliography{biblio}

\end{document}